\documentclass[12pt, reqno]{amsart}
\usepackage{amsmath, amsthm, graphicx, caption}
\usepackage[margin=2.8cm]{geometry}
\usepackage{tikz-cd}
\usepackage{algorithmic}

\usepackage{tikz}
\usepackage{tikz-cd}
\usepackage{mathtools}  
\usepackage{tikz-network}
\usepackage[colorlinks=true,linkcolor=red,citecolor=blue]{hyperref}

\newcommand{\F}{\mathbb{F}}

\newcommand{\ord}{{\rm{ord}}}

\theoremstyle{plain}
\newtheorem{theorem}{Theorem}[section]
\newtheorem{corollary}[theorem]{Corollary}

\newtheorem{proposition}[theorem]{Proposition}

\numberwithin{equation}{section}

\theoremstyle{definition}
\newtheorem{definition}[theorem]{Definition}
\newtheorem{example}[theorem]{Example}

\newtheorem{remark}[theorem]{Remark}

\newcommand{\BKLC}{{\rm{BKLC}}}

\begin{document}

\title{Equivalence of Constacyclic Codes with shift constants of different orders \ \ }

\author[dastbasteh]{Reza Dastbasteh}
\address{Tecnun, University of Navarra, San Sebastian, Spain}
\email{\url{rdastbasteh@unav.es}}

\author[padashnick]{Farzad Padashnick}
\address{Department of Pure Mathematics, University of Isfahan, Isfahan, Iran}
\email{\url{f.padashnik@sci.ui.ac.ir}}

\author[Crespo]{Pedro M. Crespo}
\address{Tecnun, University of Navarra, San Sebastian, Spain}
\email{\url{pcrespo@tecnun.es}}

\author[Grassl]{Markus Grassl}
\address{International Centre for Theory of Quantum Technologies, University of
Gdansk, Poland}
\email{\url{markus.grassl@ug.edu.pl}}

\author[Sharafi]{Javad Sharafi}
\address{Kashan University, Kashan, Iran}
\email{\url{javadsharafi@grad.kashanu.ac.ir}}

\def\keywordname{{\bf Keywords:}}

\maketitle

\begin{abstract}
Let $a$ and $b$ be two non-zero elements of a finite field
$\F_q$, where $q>2$. It has been shown that if $a$ and $b$
have the same multiplicative order in $\F_q$, then the families of
$a$-constacyclic and $b$-constacyclic codes over $\F_q$ are monomially
equivalent. In this paper, we investigate the monomial equivalence of
$a$-constacyclic and $b$-constacyclic codes when $a$ and $b$ have
distinct multiplicative orders. We present novel conditions for
establishing monomial equivalence in such constacyclic codes,
surpassing previous methods of determining monomially equivalent
constacyclic and cyclic codes.  As an application, we use these
results to search for new linear codes more systematically. In
particular, we present more than $70$ new record-breaking linear codes
over various finite fields, as well as new binary quantum codes.
\end{abstract}

{{\textit{Keywords:}} 
constacyclic code, cyclic code, monomial equivalence, isometric equivalence, new linear code.}

\section{Introduction}
Constacyclic codes are an important class of linear codes because of
their algebraic structure and efficient encoding and decoding
algorithms. Due to their similarities with cyclic codes, constacyclic
codes inherit many of the desirable properties of cyclic codes making
them well-suited for various practical applications. For instance, in
\cite{Application1}, the authors proposed a decoding algorithm for
constacyclic codes using a slight modification of that of cyclic
codes. Furthermore, as discussed in \cite[Section 1.2]{ding2022},
certain constacyclic codes can outperform cyclic codes by offering
larger minimum distances and producing Hamming codes that are not
equivalent to cyclic codes.

Let $\F_q$ be a finite field, where $q=p^m$ and $p$ is a prime
number. We denote the multiplicative group of $\F_q$ by $\F_q
^{*}$. Let $ a \in \F_q ^{*}$ and $n$ be a positive integer.  A linear
code $C$ of length $n$ over $\F_q$ is an $\F_q$ vector subspace of
$\F_q ^n$. The parameters of a linear code of length $n$, dimension
$k$, and minimum (Hamming) distance $d$ over $\F_q$ is denoted by
$[n,k,d]_q$.
We denote the multiplicative order of $a$ in $\F_q^\ast$ by $\ord(a)$.
A linear code $C$ of length $n$ over $\F_q$ is called
\textit{$a$-constacyclic} if for each codeword $(c_0,c_1,\ldots
,c_{n-1}) \in C$, we have $(a c_{n-1} ,c_0 ,c_1 ,\ldots ,c_{n-2}) \in
C$. The value $a$ is called the {\em shift constant} of such a
constacyclic code. A constacyclic code with shift constant $a=1$
is called {\em cyclic}.

The concept of equivalence of codes has been known for several
decades. It can be traced back to the early developments in coding
theory.
It involves considering two linear codes as equivalent if one can be
transformed into the other through certain operations, such as
permutation of coordinates, multiplication of coordinates by a nonzero
scalar, and applying a field automorphism to all codewords. This
concept allows us to study different codes that have the same
underlying structure and properties.

There have been several works on the classification and determination
of equivalent cyclic and constacyclic codes using algebraic properties
of these codes; for example, see \cite{Aydin1, Aydin2, Chen,
  RezaEquivalence,HP}.  In \cite{Aydin1}, the authors proved that for
each $a,b \in \F_q^\ast$ such that $\ord(a)=\ord(b)$, there exists an
isometry between the class of $a$-constacyclic and $b$-constacyclic
codes.  Therefore, the next natural step in studying equivalent
constacyclic codes is to consider codes with shift constants of
different orders.

In \cite{Bierbrauer}, J\"{u}rgen Bierbrauer established that every
$a$-constacyclic code of length $n$ over $\F_q$ is monomially
equivalent to a cyclic code of length $n$ over $\F_q$, provided that
$\gcd(\ord(a),n)=1$. This result gives a relationship between
constacyclic codes and cyclic codes, and also imposes another
restriction on constacyclic codes, allowing us to study codes that
exhibit properties that may be different from cyclic codes.

This paper extends previous works on equivalence of codes by providing
new conditions for the monomial equivalence of two constacyclic codes,
or a pair of constacyclic and cyclic codes. Specifically, we introduce
a new condition for establishing the monomial equivalence of two
constacyclic codes with shift constants of different orders.
Furthermore, we prove that our criterion generalizes the mentioned
result of Bierbrauer in \cite{Bierbrauer}, as it can detect monomially
equivalent $a$-constacyclic and cyclic codes of length $n$, even when
$\gcd(\ord(a),n)>1$. Our new results are applied to small finite
fields, where we classify families of monomially equivalent
constacyclic codes. We also use these results to search more
systematically for new linear codes. In particular, we construct many
new record-breaking linear codes over various finite fields, as well
as binary quantum codes.

\section{Preliminaries}
Throughout this paper we assume that $n$ is a positive integer such
that $\gcd(n,q)=1$.  It is well-known that there is a one-to-one
correspondence between $a$-constacyclic codes of length $n$ over
$\F_q$ and ideals of the ring $\F_q [x]/\langle x^n- a \rangle$. Thus
each $a$-constacyclic code $C$ is generated by a unique monic
polynomial $f(x)$ such that $f(x) \mid (x^n- a)$. The polynomial
$f(x)$ is called the \textit{generator polynomial} of $C$.  Therefore
to determine all $a$-constacyclic codes of length $n$ over $\F_q$ it
is essential to find the irreducible factorization of $x^n-a$.

Let $\ord(a)=t$ and $\alpha$ be a primitive $(tn)$-th root of unity in
a finite field extension of $\F_q$ such that $\alpha^n=a$. Then all
different roots of $x^n-a$ are in the form $\alpha^s$, where $s \in
\Omega _{a}$ and
\begin{equation}
\Omega _{a}=\{{kt+1}\colon 0\le k\le n-1\}.
\end{equation}

\begin{remark}
In the numerical examples of constacyclic codes throughout this paper,
the $(tn)$-th root of unity $\alpha$ is fixed as follows. First we
find the smallest integer $z$ such that $tn \mid q^z-1$.  Let $\gamma$
be the primitive element in $\F_{q^z}$ chosen by the computer
algebra system Magma \cite{magma}.  Set $\alpha_0=\gamma^{(q^z-1)/
  (tn)}$. Then there exists an integer $1\le i \le tn-1$ such that
$\alpha_0^{in}=a$. Now select $\alpha=\alpha_0^i$.
\end{remark}

For each $s \in \Omega _{a}$, the {\em $q$-cyclotomic coset} of $s$ modulo $tn$ is defined by 
\[
Z(s)=\{(q^is) \bmod{tn}\colon 0\le i \le r-1\},
\]
where $r$ is the smallest number such that $q^rs \equiv s
\pmod{tn}$. It is straight forward to see that there exist elements
$b_1,b_2,\ldots,b_m \in \Omega_a$ such that $\{Z(b_i)\colon 1\le i \le
m\}$ forms a partition of $\Omega_a$. Thus one can find the following
irreducible factorization of $x^n-a$ over $\F_q$:
\[
x^n-a=\prod_{i=1}^{m}f_i(x),
\]
where
\[
f_i(x)=\prod_{\ell \in Z(b_i)}(x-\alpha^\ell).
\]
A common approach to represent an $a$-constacyclic code with generator
polynomial $g(x)$ (the monic polynomial dividing $x^n-a$) is by its
{\em defining set} $A \subseteq \Omega_a$ such that $s \in A$ if and
only if $g(\alpha^s)=0$. It is easy to see that the defining set of an
$a$-constacyclic code is a unique (as $\alpha$ is already fixed) union
of $q$-cyclotomic cosets modulo $nt$.

An $n\times n$ matrix $M$ is called a \textit{monomial matrix} over
$\F_q$ if $M$ has exactly one nonzero entry from $\F_q$ in each row
and each column.  Let $C_1$ and $C_2$ be two linear codes of length
$n$ over $\F_q$ and $G$ be a generator matrix for $C_1$. Then $C_1$
and $C_2$ are called \textit{monomially equivalent} if there exists a
monomial matrix $M$ over $\F_q$ such that $GM$ is a generator matrix
for $C_2$.

Another commonly used notion of equivalence of
codes is based on isometries between linear codes.  We denote the
(Hamming) distance between two vectors $u$ and $v$ by $d(u,v)$.

\begin{definition}
Let $C_1$ and $C_2$ be two linear codes over $\F_q$. Then $C_1$ and
$C_2$ are called {\em isometric equivalent} if there exists an
$\F_q$-isomorphism $\phi\colon C_1 \rightarrow C_2$ such that
$d(u,v)=d(\phi(u), \phi(v))$ for each $u,v \in C_1$.
\end{definition}
The map $\phi$ with the above properties is called an {\em isometry of
  linear codes}. It was shown that the following connection holds
between monomial and isometric equivalence of codes.

\begin{theorem}\cite[Corollary 1]{Bograt}
Let $C_1$ and $C_2$ be two linear codes over a finite field. Then
$C_1$ and $C_2$ are monomially equivalent if and only if they are
isometric equivalent.
\end{theorem}

Hence the concepts of monomial and isometric equivalence are identical
in the sense of the above theorem and from now on we mainly use the
term monomial equivalence for both the monomial and isometric
equivalence of codes.

For each $a,b \in \F_q^\ast$, the families of $a$-constacyclic and
$b$-constacyclic codes of length $n$ over $\F_q$ will be called {\em
  monomially equivalent} if there exists a one-to-one correspondence
between them that is given by an isometry of linear codes. The
following result shows that constacyclic codes with shifts constants
of the same order are monomially equivalent.

\begin{theorem}\cite[Corollary 1]{Aydin1}\label{Aydin1}
 Let $n$ be a positive integer such that $\gcd (n,q)=1$. Let $a,b \in
 \F_{q}^{*}$ such that $\ord (a)=\ord (b)$. Then the families of
 $a$-constacyclic and $b$-constacyclic codes of length $n$ over $\F_q$
 are monomially equivalent.
\end{theorem}

Another result which will be generalized in the paper is the following
statement which states a sufficient condition for monomial equivalence
of cyclic and constacyclic codes.

\begin{theorem}\cite[Theorem 15]{Bierbrauer}\label{Bierbrauer}
 Let $n$ be a positive integer such that $\gcd (n,q)=1$. Let
 $a\in\F_q^\ast$ such that $\gcd (n, \ord(a))=1$. Then the families of
 $a$-constacyclic and cyclic codes of length $n$ over $\F_q$ are
 monomially equivalent.
\end{theorem}

\section{Equivalence of constacyclic codes with  shift constants of different orders}\label{S2}

In this section, we present novel results on the monomial equivalence
of constacyclic codes over the finite field $\F_q$.  To the best of
our knowledge, there has been no comprehensive discussion on the
monomial equivalence of constacyclic codes with shift constants of
different orders.  Therefore, we introduce new criteria for
establishing monomial equivalence of constacyclic codes that go beyond
the result of Theorem \ref{Aydin1}. Our result also generalizes that
of Theorem \ref{Bierbrauer} by identifying monomially equivalent
constacyclic and cyclic codes that were not previously detectable
using existing methods. These results will be applied in Section
\ref{S:examples} to identify new record-breaking linear codes.

 For the rest of this section, we fix $\xi$ to be a generator of the
 multiplicative group $\F_q^\ast$, i.e.,
 \[
 \F_q ^{*}=\langle \xi \rangle.
 \] 
Our main result of this section is given below which discusses
monomial equivalence of constacyclic codes with shift constants of
not necessarily the same order.
 
\begin{theorem}\label{mainTheorem}
  Let $\xi ^{i_1}, \xi^{i_2} \in \F_q^\ast$ for some $0\le i_1,i_2 \le
  q-2$ such that $i_1=i_2 s$ for some non-negative integer $s$. Let $n$ be
  a positive integer such that $\gcd (n,q)=1$ and $\gcd (n,q-1)=m$. If
  $m \mid \gcd (i_2 (s-1),q-1)$, then the families of
  $\xi^{i_1}$-constacyclic and $\xi^{i_2}$-constacyclic codes of
  length $n$ over $\F_q$ are monomially equivalent.
\end{theorem}

\begin{proof}
Let $n=m \beta$ for some positive integer $\beta$, where $\gcd (\beta
,q-1)=1$. The fact that $m\mid\gcd (i_2 (s-1),q-1)$ implies that there
exist integers $\gamma$ and $\theta$ such that $m \gamma
=i_2 (s-1)$ and $m \theta=q-1$.  Since $\gcd (\beta ,q-1)=1$, we have
$\gcd (\beta ,\theta)=1$. So there exists $1\le \beta' \le \theta-1$
such that $\beta \beta ' \equiv 1 \pmod \theta$.  Next, we show that
the map
\[
\Phi\colon\F_q [x]/ \langle x^n -\xi ^{i_1} \rangle \longrightarrow\F_q [x] / \langle x^n -\xi ^{i_2} \rangle
\]
defined by $\Phi (x)= \xi ^i x$, where $i=\gamma \beta' +\theta$, is an isometry of linear codes. 
First we show that it is well-defined. It is sufficient to show that
\begin{align*}
\Phi (x^n -\xi ^{i_1}) =\xi ^{in}x^{n}-\xi ^{i_1} = \xi ^{in}(x^n -\xi ^{i_1 -in}) \in \langle x^n -\xi ^{i_2} \rangle.
\end{align*}
Note that since $\beta \beta ' \equiv 1 \pmod{\theta}$, there exists
an integer $\ell$ such that $\beta \beta ' =1+\ell\theta$. Using the
facts that $i_1=i_2 s$, $i=\gamma\beta'+\theta$, $n=m\beta$,  $m\gamma
=i_2 (s-1)$, and $m\theta=q-1$, we compute
\begin{equation}
\begin{split}
i_1 - in &= i_1 -(\gamma\beta'+\theta)m\beta =  i_2s- \gamma m (1+\ell\theta) = i_2s- i_2(s-1) - (q-1)\gamma \ell  \\
& \equiv i_2 \pmod{q-1}.
\end{split}
\end{equation}
This shows that $\Phi (x^n -\xi ^{i_1}) \in \langle x^n -\xi ^{i_2}
\rangle$ and therefore $\Phi$ is well-defined. The map $\Phi$ is an
evaluation map, so it is a ring homomorphism, and it preserves the
weight of vectors. Thus it is enough to show that $\Phi$ is
surjective, and the fact that it is a map between finite sets implies
that $\Phi$ is a ring isomorphism. Let $p(x) \in \F_q [x] / \langle
x^n -\xi ^{i_2} \rangle$. Then $p'(x)=p(\xi^{-i}x) \in \F_q [x]/
\langle x^n -\xi ^{i_1} \rangle$ and $\Phi(p'(x))=p(x)$. This proves
that $\Phi$ is surjective. Thus it is a ring isomorphism that
preserves the distance between each pair of vectors.
\end{proof}

Note that in Theorem \ref{mainTheorem} the condition $i_1=i_2 s$
implies that $\ord(\xi ^{i_1}) \mid \ord (\xi ^{i_2})$.  The next
corollary gives a special case of Theorem \ref{mainTheorem} which can
be easily applied to check the monomial equivalence of certain
constacyclic codes.

\begin{corollary}\label{GeneralizedAydingcd1}
Let $a,b \in \F_q ^{*}$ such that $\ord (a) \mid \ord (b)$. Let $n$ be
a positive integer such that $\gcd (n,q)=\gcd (n,q-1)=1$. Then the
families of $a$-constacyclic and $b$-constacyclic codes of length $n$
over $\F_q$ are monomially equivalent.
\end{corollary}

\begin{proof}
Let $a=\xi^{i_1}$ and $b=\xi^{i_2}$ for some $0\le i_1,i_2 \le
q-2$. Then $\ord(a)\mid\ord(b)$ implies that there exists a
non-negative integer $s$ such that $i_1=i_2s$.  Now the proof follows
from Theorem \ref{mainTheorem} since $1 \mid \gcd(i_2(s-1),q-1)$.
\end{proof}

Next, we provide an example of monomially equivalent constacyclic
codes that cannot be obtained by Theorem \ref{Aydin1}. Nevertheless,
such cases can be justified using the outcome of Corollary
\ref{GeneralizedAydingcd1}.

\begin{example}
Let $a=4$ and $b=2$ be elements of $\F_5$. We have $\ord (2)=4$ and
$\ord (4)=2$ and hence $\ord (4) \mid \ord (2)$. By Corollary
\ref{GeneralizedAydingcd1} we have that $2$-constacyclic and
$4$-constacyclic codes of length $n$ over $\F_5$ are monomially
equivalent for each $n$ satisfying $\gcd(n,5)=\gcd(n,4)=1$.
\end{example}

Next, we provide another example of monomially equivalent constacyclic
codes that can be detected by Theorem \ref{mainTheorem}, but not by
Corollary \ref{GeneralizedAydingcd1}. This stresses that the result of
Corollary \ref{GeneralizedAydingcd1} is only a special case of Theorem
\ref{mainTheorem},  although Corollary \ref{GeneralizedAydingcd1} can
be applied easier in practice.

\begin{example}
Let $a=6$ and $b=5$ be elements of $\F_7^\ast=\langle 5 \rangle$. We
have $6=5^3$ in $\F_7$. Let $n$ be a positive integer such that $\gcd
(n,7)=1$ and $\gcd (n,6)=2$. Let $i_1=3$, $i_2=1$, and $s=3$
($a=5^{i_1}$, $b=5^{i_2}$, and $i_1=i_2s$). Note that  $2 \mid
\gcd(i_2(3-1),6)=\gcd (2,6)$.  Hence Theorem \ref{mainTheorem} shows
that $a$-constacyclic and $b$-constacyclic codes of length $n$ over
$\F_7$ are monomially equivalent. On the other hand, Corollary
\ref{GeneralizedAydingcd1} requires $\gcd(n,7)=\gcd(n,6)=1$.
\end{example}

Recall that in a cyclic group of order $m$ generated by an element
$a$, we have $\ord(a^k)=\frac{m}{\gcd(m,k)}$ for any $0\le k \le
m-1$. Next we prove that Theorem \ref{Bierbrauer} is a special case of
Theorem \ref{mainTheorem}.

Let $C$ be an $a$-constacyclic code over $\F_q$ of length $n$ such
that $\gcd(n,q)=1$ and $\gcd(n,\ord(a))=1$. Suppose that $a=\xi^i$ for
some $0\le i \le q-2$. Since $\gcd(n,\ord(a))=1$ and
$\ord(a)=\frac{q-1}{\gcd(i,q-1)}$, we have $\gcd(n,
\frac{q-1}{\gcd(i,q-1)})=1$. This implies that
$\gcd(n,q-1)=\gcd(n,\gcd(i,q-1))$.  Note also that the latter equality
implies that $\gcd(n,q-1) \mid \gcd(i,q-1)$.  Now $\xi^i$ and $\xi^0$
satisfy the conditions of Theorem \ref{mainTheorem}. In other words,
putting $i_2=i$ and $i_1=0$ implies that $\gcd(n,q-1) \mid
\gcd(i(0-1),q-1)=\gcd(i,q-1)$.  So Theorem \ref{mainTheorem} results
that the families of $a$-constacyclic and cyclic codes of length $n$
over $\F_q$ are monomially equivalent.

In the following example, we show that Theorem \ref{mainTheorem}
offers a wider range of applicability than Theorem \ref{Bierbrauer} in
determining monomially equivalent constacyclic and cyclic codes.

\begin{example}\label{ex.F5Bier}
In the field $\F_5$, we have $\F_5^\ast=\langle 2\rangle$.  Let $n$ be
a positive integer in the form $n=4k+2$ such that $\gcd (n, 5)=1$. We
show the families of $4$-constacyclic codes and cyclic codes of length
$n$ over $\F_5$ are monomially equivalent.  First note that
$\ord(4)=2$ and $\gcd (n, \ord(4))=2$. Therefore, we cannot apply
Theorem \ref{Bierbrauer} in this case. Moreover, we have $\gcd
(n,5-1)=2$. Applying the result of Theorem \ref{mainTheorem} to
$2^{i_2}$ and $2^{i_1}$ for $i_2=2$ and $i_1=0$ implies that
$4$-constacyclic and cyclic codes of length $n$ over $\F_5$ are
monomially equivalent.
\end{example}

This example can be generalized as follows. 

\begin{corollary}\label{C:3.6}
Let $n$ and $m$ be positive integers such that $\gcd(n,q-1)=m$. Then the families of $\xi^{mr}$-constacyclic and cyclic codes of length $n$ over $\F_q$ are monomially equivalent, for any positive integer $r$.
\end{corollary}

\begin{proof}
Let $i_2=mr$ and $i_1=0$. Next we apply Theorem \ref{mainTheorem} to $\xi^{i_2}$ and $\xi^{i_1}$. First note that $\gcd(n,q-1)=m$ and $m\mid \gcd(i_2(0-1),q-1)$. Now all the conditions of Theorem \ref{mainTheorem} are satisfied and thus the result follows from it. 
\end{proof}

By Corollary \ref{C:3.6}, if $\gcd(n,q-1)=m$, then $\xi^{mr}$-constacyclic and cyclic codes of length $n$ over $\F_q$ are monomially equivalent for any positive integer $r$. Suppose additionally that $\gcd(r,m)=1$ and $m^2 \mid q-1$. Then $m \mid \frac{q-1}{\gcd(mr,q-1)}$ yielding $m\mid \ord(\xi^{mr})$. As $m \mid n$, we have $m \mid \gcd(n,\ord(\xi^{mr}))$ implying that  Theorem \ref{Bierbrauer} cannot be applied to prove monomial equivalence  when $\gcd(r,m)=1$ and $m^2 \mid q-1$.

\section{Classification of equivalent constacyclic codes over Small Finite Fields}\label{S3}

In this section, we study the restriction of the results in the
previous section to constacyclic codes over small finite fields when
$q\le 7$. This illustrates our new results and
also helps to classify families of monomially equivalent constacyclic
codes over small finite fields. In the next section, in our computer
search for new linear codes, we will apply the results of this section
to speed up the search. In fact we give at least one new
record-breaking linear code over each of the finite fields discussed
in this section.
 
We use graphs to depict the monomial equivalence of constacyclic
codes. In particular, for each finite field $\F_q$, we use a graph in
which a vertex labeled with $a$ corresponds to the family of
$a$-constacyclic codes of length $n$ over $\F_q$. Moreover, each edge
in the graph represents the condition under which two families of
constacyclic codes are monomially equivalent. Such a graph will be
called the ``equivalence graph" of constacyclic codes.
Recall that we only consider constacyclic codes of length
$n$ over $\F_q$ when $\gcd(n,q)=1$.

The first nontrivial constacyclic codes appear over $\F _3$. By
Theorem \ref{Bierbrauer}, if $\gcd (n,2)=1$, the families of
cyclic and $2$-constacyclic codes of length $n$ over $\F_3$ are
monomially equivalent (see Figure \ref{F1}).
\begin{figure}[h!]
\begin{center}
\begin{tikzpicture} 
\Vertex[label=$1$, size=0.7, fontscale=1 , opacity =.2 ]{1}
\Vertex[x=4, label=$2$, size=0.7, fontscale=1 , opacity =.2]{2}
\Edge[color=black, label={\ $\gcd(n,2)=1$ \ \ }](1)(2)
\end{tikzpicture}
\caption{Monomially equivalent $a$-constacyclic codes of length $n$ over $\F_3$ with $a\in \{1,2\}$.}
\label{F1}
\end{center}
\end{figure}
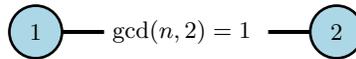

The next example shows that if $\gcd (n,2)\neq 1$, then cyclic and
$2$-constacyclic codes of length $n$ over $\F_3$ are not necessarily
monomially equivalent.

\begin{example}
Let $n=10$. Then $\gcd (10,3)=1$ and $\gcd (10,2)=2$. The irreducible
decompositions of $x^{10}-1$ and $x^{10}-2$ are
\[
x^{10}-1=(x + 1)(x + 2)(x^4 + x^3 + x^2 + x + 1)(x^4 + 2x^3 + x^2 + 2x + 1)
\]
and 
\[
x^{10}-2=(x^2 + 1)(x^4 + x^3 + 2x + 1)(x^4 + 2x^3 + x + 1).
\]
Therefore, the number of cyclic codes of length $10$ over $\F_3$ is
equal to $16$, but the number of $2$-constacyclic codes of length $10$
over $\F_3$ is equal to $8$. Hence the families of cyclic and
$2$-constacyclic codes of length $10$ over $\F_3$ are not monomially
equivalent.
\end{example}

Let $\F_4=\{0,1,\omega,\omega^2\}$, where $\omega^2=\omega+1$. Note
that $\omega$ and $\omega^2$ have the same multiplicative order. So by
Theorem \ref{Aydin1}, the families of $\omega$-constacyclic and
$\omega^2$-constacyclic codes of length $n$ over $\F_4$ are monomially
equivalent. Moreover, Theorem \ref{Bierbrauer} implies that the
families of cyclic and $\omega$-constacyclic codes of length $n$ over
$\F_4$ are monomially equivalent provided that $\gcd(n,3)=1$. Figure
\ref{F2} gives the equivalence graph of constacyclic codes over
$\F_4$. We use the notation $a\approx b$ to say that the families of
$a$-constacyclic and $b$-constacyclic codes of length $n$ are
monomially equivalent.

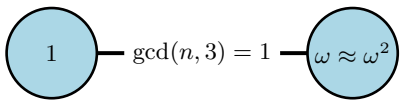
\begin{figure}[h!]
\begin{center}

\begin{tikzpicture} 
\Vertex[label=$1$, size=1.2, fontscale=1 , opacity =.2 ]{1}
\Vertex[x=4, label=$\omega\approx \omega^2$, size=1.2, fontscale=1 , opacity =.2]{2}
\Edge[ color=black, label={\ $\gcd(n,3)=1$\  \ }](1)(2)
\end{tikzpicture}
\caption{Monomially equivalent $a$-constacyclic codes of length $n$ over $\F_4$ with $a\in \{1,\omega,\omega^2\}$.}
\label{F2}
\end{center}

\end{figure}

Next we consider constacyclic codes over $\F_5$. First note that
$\ord(2)=\ord(3)=4$ and $\ord(4)=2$. Thus the families of
$2$-constacyclic and $3$-constacyclic codes over $\F_5$ are monomially
equivalent by Theorem \ref{Aydin1}. Moreover, Theorem \ref{Bierbrauer}
implies that if $\gcd(n,2)=1$, then
 \begin{itemize}
\item families of cyclic and $2$-constacyclic codes of length $n$ over $\F_5$ are monomially equivalent.
\item families of cyclic and $4$-constacyclic codes of length $n$ over $\F_5$ are monomially equivalent.
\end{itemize}
Next we give new conditions for monomial equivalence of constacyclic codes over $\F_5$ which cannot be deduced from
Theorems \ref{Aydin1} and \ref{Bierbrauer}.

\begin{proposition}\label{Thm F_5 (2,4)}
Let $\gcd (n,5)=\gcd (n,2)=1$, then $2$-constacyclic and $4$-constacyclic codes of length $n$ over $\F_5$ are monomially equivalent.
\end{proposition}

\begin{proof}
Note that $\ord (4) \mid \ord (2)$ and thus the claim follows from Corollary \ref{GeneralizedAydingcd1}.
\end{proof}

It should be noted that if, for example, $\gcd (n,5)=1$ and $\gcd
(n,2)\neq 1$, then the families of $2$-constacyclic and
$4$-constacyclic codes of length $n$ over $\F_5$ are not necessarily
monomially equivalent, as demonstrated next.
\begin{example}
Let $n=12$. Then $\gcd (n,5)=1$ and $\gcd (n,2)=2$.
The total number of $2$-constacyclic codes of length $12$ over $\F_5$ is $8$, but the number of $4$-constacyclic codes of length $12$ over $\F_5$ is $64$. 
Therefore, the families of $2$-constacyclic and $4$-constacyclic codes of length $12$ over $\F_5$ are not monomially equivalent.
\end{example}

Hence, the $\gcd$ condition of Proposition \ref{Thm F_5 (2,4)} cannot
be improved for all values of $n$.  Note also that, as shown in
Example \ref{ex.F5Bier}, if $\gcd(n,4)=2$, then the families of
$4$-constacyclic and cyclic codes of length $n$ over $\F_5$ are
monomially equivalent. The equivalence graph of constacyclic codes
over $\F_5$ is given in Figure \ref{F3}.

\begin{figure}[h!]
\begin{center}
\begin{tikzpicture} 
\Vertex[label=$1$, size=1, fontscale=1, opacity =.2 ]{1}
\Vertex[x=3,y=3, label=$2\approx 3$, size=1, fontscale=1 , opacity =.2]{2}
\Vertex[x=6,label=$4$, size=1 , fontscale=1 , opacity =.2]{3}
\Edge[Math, color=black, label={\hbox to 20pt{\hss$\gcd(n,2)=1$\hss}}, distance=.5](1)(2)
\Edge[Math, color=black, label={\hbox to 20pt{\hss$\gcd(n,2)=1$\hss}}, distance=.5](2)(3)
\Edge[Math, color=black, label={\ \gcd(n,4)\in\{1,2\} \ }](3)(1) 
\end{tikzpicture}
\caption{Monomially equivalent $a$-constacyclic codes of length $n$ over $\F_5$ with $a\in \{1,2,3,4\}$.}
\label{F3}
\end{center}
\end{figure}
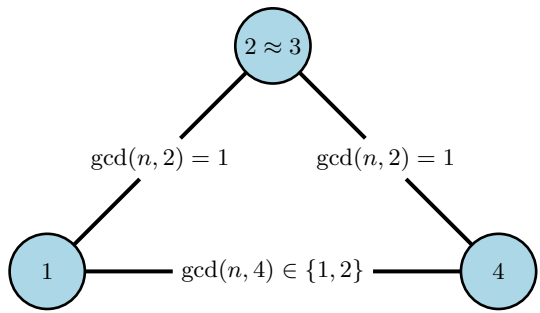

Next we study constacyclic codes over $\F_7$ of length $n$ such that $\gcd(n,7)=1$. First note that $\ord(3)=\ord(5)=6$, $\ord(2)=\ord(4)=3$, and $\ord(6)=2$.
So by Theorem \ref{Aydin1}, the families of 
\begin{itemize}
\item $2$-constacyclic and $4$-constacyclic codes 
\item  $3$-constacyclic and $5$-constacyclic codes 
\end{itemize}
of length $n$ over $\F_7$ are monomially equivalent. Moreover, by
Theorem \ref{Bierbrauer}, the following families of codes are
monomially equivalent:
\begin{itemize}
\item cyclic and $2$-constacyclic codes of length $n$ over $\F_7$ when $\gcd(n,3)=1$;
\item cyclic and $3$-constacyclic codes of length $n$ over $\F_7$ when $\gcd(n,6)=1$;
\item cyclic and $6$-constacyclic codes of length $n$ over $\F_7$ when $\gcd(n,2)=1$.
\end{itemize}
Now, we provide new conditions for establishing monomial equivalence
of constacyclic codes over $\F_7$, which cannot be deduced from
Theorems \ref{Aydin1} and \ref{Bierbrauer}.

\begin{proposition}\label{propF_7}
Let $n$ be a positive integer such that $\gcd(n,7)=1$. Then the following families are monomially equivalent:
\begin{itemize}
\item[$(i)$]
 $2$-constacyclic and $3$-constacyclic codes of length $n$ over $\F_7$ when $\gcd(n,2)=1$;
\item[$(ii)$]
$3$-constacyclic and $6$-constacyclic codes of length $n$ over $\F_7$ when $\gcd(n,3)=1$.
\end{itemize}
\end{proposition}

\begin{proof}
There are three cases to consider: when $\gcd(n,6)=1, 2$, or $3$.
If $\gcd(n,6)=1$, then both $(i)$ and $(ii)$ follow from Corollary
\ref{GeneralizedAydingcd1} as $\ord(2)=3 \mid \ord(3)=6$ and
$\ord(6)=2 \mid\ord(3)=6$. 

Let $\gcd(n,6)=2$. Note that $6=3^{i_1}$ and $3=3^{i_2}$ over $\F_7$,
where $i_1=3$ and $i_2=1$. Moreover, $2 \mid \gcd(i_2(3-1),6)=2$. So
all the conditions of Theorem \ref{mainTheorem} are satisfied. Hence
$3$-constacyclic and $6$-constacyclic codes of length $n$ over $\F_7$
are monomially equivalent.

Let $\gcd(n,6)=3$. We have $4=3^{i_1}$ and $3=3^{i_2}$ over $\F_7$,
where $i_1=4$ and $i_2=1$. Moreover, $3 \mid
\gcd(i_2(4-1),6)=3$. Therefore, all the conditions of Theorem
\ref{mainTheorem} are satisfied and $3$-constacyclic and
$4$-constacyclic codes of length $n$ over $\F_7$ are monomially
equivalent. Now the fact that families of $2$-constacyclic and
$4$-constacyclic codes over $\F_7$ are monomially equivalent completes
the proof.
\end{proof}

The equivalence graph of constacyclic codes over $\F_7$ is provided in
Figure \ref{F4}. Similar to previous cases, one can show through
examples that the given gcd conditions are sharp, i.e., relaxing them
may result in inequivalent constacyclic codes for certain lengths.

\begin{figure}[h!]
\begin{center}
\begin{tikzpicture} 
\Vertex[label=$1$, size=1, fontscale=1 , opacity =.2]{1}
\Vertex[x=4.5, label=$2\approx 4$, size=1, fontscale=1 , opacity =.2]{2}
\Vertex[y=-4.5, label=$6$, size=1 , fontscale=1 , opacity =.2]{3}
\Vertex[x=4.5 ,y=-4.5, label=$3\approx 5$, size=1 , fontscale=1, opacity =.2]{4}
\Edge[Math, color=black, label={\ \gcd(n,3)=1\ }](1)(2)
\Edge[Math, color=black, label={\hbox to 10pt{\hss$\gcd(n,2)=1$\hss}}, distance=.75](1)(3)
\Edge[Math, color=black, label={\hbox to 15pt{\hss$\gcd(n,6)=1$\hss}}](1)(4)
\Edge[Math, color=black, label={\hbox to 10pt{\hss$\gcd(n,2)=1$\hss}}, distance=0.25](2)(4) 
\Edge[Math, color=black, label={\ \gcd(n,3)=1\ }](3)(4) 
\end{tikzpicture}
\caption{Monomially equivalent $a$-constacyclic codes of length $n$ over $\F_7$ with $a\in \{1,2,\cdots, 6\}$.}
\label{F4}
\end{center}
\end{figure}

\section{New record-breaking codes}\label{S:examples}

Finding linear codes with good parameters is a challenging task in
algebraic coding theory. Many works in the literature have made
efforts to develop systematic approaches for computer-based searches
of good codes. However, the computational complexity associated with
computations like the code's minimum distance has significantly slowed
down the search process. Recently, advancements have been made by
employing the equivalence of linear cyclic and constacyclic codes to
design more efficient search algorithms
\cite{Aydin1,Aydin2,RezaEquivalence}. As a result, dozens of new
linear codes and binary quantum codes have been discovered.

In this section, we utilize the results on the equivalence of
constacyclic codes to cut down the size of the search space for new
linear codes with good parameters. To begin, we employ the results
presented in Theorems \ref{Aydin1} and \ref{Bierbrauer}, as well as
Sections \ref{S2} and \ref{S3}, to identify all families of monomially
equivalent constacyclic codes with distinct shift constants. Once this
is done, we focus on a specific shift constant $a$ and apply the
results discussed in \cite{akre2023,RezaEquivalence} to pinpoint
monomially equivalent $a$-constacyclic codes.  This helps to search
more systematically for new linear codes with good parameters.
In certain scenarios, calculating the actual minimum distance can be a
time-consuming process, sometimes taking several days or even
years. However, by identifying equivalent codes and narrowing down the
search space, the computation can be restricted to a smaller set of
codes.
 
In particular, by pruning the search space using these results we were
able to identify more than $70$ new record-breaking linear codes over
various finite fields and binary quantum codes. In our search we
considered the following methods:
 \begin{itemize}
 \item direct construction of constacyclic codes, 
 \item applying Construction X to two constacyclic (cyclic) codes with an auxiliary best-known linear code,
 \item applying Construction XX to three constacyclic (cyclic) codes with two auxiliary best-known linear codes.
 \end{itemize}
 
In our computations we represent the currently best-known linear code
of length $n$ and dimension $k$ by $\BKLC(n,k)$. The parameters of
such codes are available online \cite{Grassl}.

Before presenting our new codes, we first recall Construction X (see
\cite[Chapter 18, \S7]{MS}). Let
$C_1$ and $C_2$ be linear codes over $\F_q$ with parameters $[n, k_1,
  d_1]_q$ and $[n, k_2, d_2]_q$ such that $C_2 \subseteq C_1$. Then,
we can express $C_1$ as union of $q^{k_1-k_2}$ different cosets of
$C_2$, namely $D_i$ for $1\le i \le q^{k_1-k_2}$, with $D_1=C_2$.

Let $C_3$ be another linear code over $\F_q$ with parameters
$[n_3,k_3,d_3]_q$ such that $k_3=k_1-k_2$. We denote each element of
$C_3$ by $z_i$ for $1\le i \le q^{k_3}$, with $z_1=0$.  {\em Construction X} of
$(C_1,C_2,C_3)$ is defined by
\begin{equation}\label{CX}
E=\{(x_i,z_i)\in \F_q^{n+n_3}\colon x_i \in D_i, 1\le i\le q^{k_1-k_2} \}.
\end{equation}
The labeling of the cosets $D_i$ and the vectors $z_i\in C_3$ can
chosen such that the new code $E$ is also a linear code over $\F_q$
and has parameters $[n + n_3, k_1,d \geq \min (d_2, d_1 + d_3) ]_q$.
To apply Construction X to linear codes $(C_1,C_2,C_3)$, we use the
built-in function in Magma \cite{magma} using the command
\texttt{ConstructionX(C1, C2, C3)}.  Our numerical results result in
several improvements of the parameters of codes in the online tables
\cite{Grassl}.

\begin{example}\label{EN:1}
  Let $n=145$. As Figure \ref{F1} shows, the fact that $\gcd(145,2)=1$
  implies that cyclic and $2$-constacyclic codes of length $145$ are
  monomially equivalent over $\F_3$. So we can restrict the
  search to cyclic codes of length $145$.  Let $C$ be the cyclic
  code of length $145$ with the defining set $Z(0) \cup Z(1)$ over
  $\F_3$.  Our computations using Magma show that $C$ has parameters
  $[145,116,10]_3$ which is a new {\em record-breaking} linear code
  over $\F_3$. The previously best-known code with the same length and
  dimension had minimum distance $9$.  Shortening of this code gives
  four other record-breaking codes over $\F_3$ with parameters
  $[145-i,116-i,10]_3$ for each $1\le i \le 4$.
\end{example}

\begin{example}\label{EN:2}
  Let $n=39$ and $C_1$ and $C_2$ be the length $39$
  $\omega$-constacyclic code over $\F_4$ with the defining sets
  $Z(10)\cup Z(19)$ and $Z(10)\cup Z(13)\cup Z(19)$, respectively. Our
  computations show that $C_1$ and $C_2$ have parameters $[39,27,7]_4$
  and $[39,24,9]_4$, respectively and $C_2\subseteq C_1$.  Let
  $C_3=[3,3,1]_4$.  Now applying Construction X to $(C_1,C_2,C_3)$
  gives a linear code $D$ with parameters $[42,27,d\ge 8]_4$ over
  $\F_4$. Computing the true minimum distance of $D$ showed that
  $d=9$. Therefore $D$ is a {\em record-breaking} linear code as the
  previously best-known linear code with the same length and dimension
  over $\F_4$ had minimum distance $8$. Shortening of this code gives
  two other record-breaking codes over $\F_4$ with the parameters
  $[41,26,9]_4$ and $[40,25,9]_4$.
\end{example}

\begin{example}\label{EN:3}
 Let $n=109$. Since $\gcd(109,3)=1$, as Figure \ref{F2} indicates, the
 families of cyclic and constacyclic codes over $\F_4$ are monomially
 equivalent.  So we can restrict our search to cyclic codes of length
 $109$.  Let $C$ be the cyclic code of length $109$ with defining set
 $Z(1)\cup Z(3)$ over $\F_4$.  Our computation using Magma reveals
 that the linear code $C$ has parameters of $[109,73,16]_4$,
 establishing that it is a new {\em record-breaking} linear code over
 $\F_4$. Computing the minimum distance took a bit more than $5$ CPU
 years and was completed in about using $80$ hours real time using up
 to $600$ cores.

 Using standard procedures like shortening, puncturing, lengthening,
 and subcodes, we find $56$ additional codes with improved parameters.
 Moreover, we can derive more than $20$ other codes in \cite{Grassl}
 for which the lower bound on the minimum distance was derived by
 non-constructive methods such as the Gilbert–Varshamov bound.

 The code $C$ contains its Hermitian dual. Applying the quantum
 stabilizer construction of \cite[Theorem 2]{Calderbank} to it implies
 the existence of a $[\![109,37,16]\!]_2$ binary quantum code. This is
 a {\em record-breaking quantum code} as the previous code with the
 same length and dimension had minimum distance $14$ in the database of
 currently best-known binary quantum codes \cite{Grassl}.
\end{example}

\begin{example} \label{EN:4}
  (1) Let $n = 87$. The equivalence graph in Figure \ref{F3} and the
  fact that $\gcd(87, 2) = 1$ imply that cyclic and $a$-constacyclic
  codes over $\F_5$ of length $87$ are monomially equivalent for $a
  \in {2, 3, 4}$. Let $C_1$, $C_2$, and $C_3$ be linear cyclic codes
  of length $87$ over $\F_5$ with defining sets $Z(0)\cup Z(2)\cup
  Z(3) \cup Z(4)$, $Z(2)\cup Z(3) \cup Z(4) \cup Z(29)$, and $Z(0)\cup
  Z(2)\cup Z(3) \cup Z(4) \cup Z(29)$, respectively.  Our computations
  using Magma yield that these codes have parameters $[87,44,22]_5$,
  $[87,43,23]_5$, and $[87,42,24]_5$, respectively, all with a {\em
    larger minimum distance} than the currently best-known codes of
  the same length and dimension.  The calculation of the minimum
  distance of the code $[87,42,24]_5$ took a bit more than $1000$ CPU
  days.

  Shortening these codes results in record-breaking codes with parameters 
  \begin{itemize}
  \item $[87 - i_1, 44 - i_1, 22]_5$ for $1 \leq i_1 \leq 2$, 
  \item $[87 - i_2, 43 - i_2, 23]_5$ for $1 \leq i_2 \leq 7$, and
  \item $[87 - i_3, 42 - i_3, 24]_5$ for $1 \leq i_3 \leq 6$.
  \end{itemize}
  Moreover, applying Construction X to $(C_2, C_3, \text{BKLC}(1, 1))$
  and $(C_1, C_3, \text{BKLC}(3, 2))$ yields codes with parameters
  $[88,43,24]_5$ and $[90,44,24]_5$, respectively, surpassing the
  currently best-known codes with the same length and dimension over
  $\F_5$. 

  (2) Let $n=101$. As Figure \ref{F3} shows, the fact that
  $\gcd(101,2)=1$ implies that cyclic and $a$-constacyclic codes over
  $\F_5$ of length $101$ are monomially equivalent for $a \in
  \{2,3,4\}$. So it is sufficient to consider cyclic codes over $\F_5$
  of length $101$.  Let $C$ be the cyclic code of length $101$ with
  the defining set $Z(0)\cup Z(1)$ over $\F_5$.  Using Magma, it took
  about $44.5$ CPU days to show that the linear code $C$ has
  parameters $[101,75,13]_5$, which is a new {\em record-breaking}
  linear code over $\F_5$.  This code surpasses the previously
  best-known linear code of the same length and dimension which had
  minimum distance $12$. Shortening of this code produces $22$ other
  record-breaking linear codes over $\F_5$ with the parameters
  $[101-i,75-i,13]_5$ for each $1\le i \le 22$.
\end{example}

Before presenting our last examples, we briefly recall Construction XX
\cite{XX} which is a generalization of Construction X.

Let $C$ be an $[n, k, d]_q$ linear code over $\F_q$ containing two
subcodes $C_1$ and $C_2$ with parameters $[n,k - k_1,d_1]_q$ and
$[n,k- k_2,d_2]_q$, respectively. Let $D_1$ and $D_2$ be two linear
codes with parameters $[n_1,k_1,\delta_1]_q$ and
$[n_2,k_2,\delta_2]_q$, respectively.  Then appending vectors from
$D_1$ and $D_2$ to the end of the cosets of $C$ with respect to $C_1$
and $C_2$ (which form a subspace of the direct product $C\times D_1
\times D_2$) implies the existence of an
\[
[n+n_1+n_2, k, \min(\delta_0, d_1 + \delta_1, d_2 + \delta_2, d + \delta_1 + \delta_2)]_q
\]
linear code over $\F_q$, where $\delta_0=d(C_1 \cap C_2)$. We take
advantage of the built-in function \texttt{ConstructionXX(C1, C2, C3,
  D1, D2)} in Magma \cite{magma} to form Construction XX of linear
codes $(C_1,C_2,C_3)$ using the auxiliary codes $D_1$ and $D_2$.

In the following example, we construct four {record-breaking} linear
codes over $\F_7$ using Construction X and XX.

\begin{example}\label{EN:5}
(1) Let $n=57$ and $C_1$ and $C_2$ be the cyclic codes of length $57$
  over $\F_7$ with the defining sets $Z(1)\cup Z(5)\cup Z(12)\cup
  Z(22)$ and $Z(0) \cup Z(1)\cup Z(5)\cup Z(12)\cup Z(22)$,
  respectively.  Then $C_2 \subseteq C_1$. Our computations using Magma
  show that $C_1$ and $C_2$ have parameters $[57,45,7]_7$ and
  $[57,44,8]_7$, respectively. Applying Construction X to
  $(C_1,C_2,\BKLC(1,1))$ gives a {\em record-breaking} linear codes
  over $\F_7$ with the parameters $[58,45,8]_7$.

(2) Let $n=58$ and $C_1$, $C_2$, $C_3$, and $C_4$ be the linear cyclic
  codes of length $58$ over $\F_7$ with the defining sets $Z(1)\cup
  Z(2)$, $Z(0)\cup Z(1)\cup Z(2)$, $Z(1)\cup Z(2)\cup Z(29)$, and
  $Z(0) \cup Z(1)\cup Z(2)\cup Z(29)$, respectively.  The following
  inclusions holds between the mentioned codes $C_4 \subset C_2
  \subset C_1$ and $C_4 \subset C_3 \subset C_1$. Moreover, $C_4=C_2
  \cap C_3$. Our computations using Magma show that $C_1$, $C_2$, $C_3$,
  and $C_4$ have parameters $[58,44,8]_7$, $[58,43,9]_7$,
  $[58,43,9]_7$, and $[58,42,10]_7$, respectively.
  Applying Construction XX to $(C_1,C_2,C_3)$ with the auxiliary codes
  $D_1=D_2=[1,1,1]_7$ implies the existence of a code $[60,44,10]_7$,
  which is a new {\em record-breaking} linear code over $\F_7$.

(3) Let $n=74$ and and $C_1$, $C_2$, $C_3$, and $C_4$ be the linear
  cyclic codes of length $74$ over $\F_7$ with the defining sets
  $\mathcal{Z}_0=Z(4)\cup Z(5) \cup Z(10)\cup Z(15)$ ,
  $\mathcal{Z}_0\cup Z(37)$, $\mathcal{Z}_0\cup Z(0)$, and
  $\mathcal{Z}_0\cup (0) \cup Z(37)$, respectively. We have
  $C_4\subset C_2\subset C_1$, $C_4\subset C_3\subset C_1$, and
  $C_4=C_2\cap C_3$. Our computations using Magma show that $C_1$,
  $C_2$, $C_3$, and $C_4$ have parameters $[74,38,20]_7$,
  $[74,37,21]_7$, $[74,37,21]_7$, and $[74,36,22]_7$, respectively.
  All these codes improve the lower bound on the minimum distance.
  Applying Construction XX to $(C_1,C_2,C_3)$ with the auxiliary codes
  $D_1=D_2=[1,1,1]_7$ implies the existence of a code $[76,38,22]_7$,
  which is a new {\em record-breaking} linear code over $\F_7$. We do
  not explictly list the additional codes that can be trivially
  derived from these new codes.
\end{example}

Table \ref{T1} lists the parameters of our new linear codes that were
discussed in this section.
\medskip

\begin{table}[h!]
  \centering
      \begin{tabular}{|p{0.31cm}|p{7.6cm}|p{2.55cm}|}
      \hline
{\centering$q$}    & \centering{Parameters} & {\centering{Reference}} \\
          \hline
$3$ &$[145-i,116-i,10]_3 \hfill 0\le i \le 4$ & Example \ref{EN:1}\\
\hline
$4$ & $[42-i,27-i,9]_4 \hfill  0 \le i \le 2$ & Example \ref{EN:2}\\
$4$ & $[109-i,70-i,16]_4 \hfill  0\le i \le 3$ & Example \ref{EN:3}\\
$4$ & $[109-i,71-i,16]_4 \hfill 0\le i\le 12$ & Example \ref{EN:3}\\
$4$ & $[109-i,72-i,16]_4 \hfill0\le i \le 12$ & Example \ref{EN:3}\\
$4$ & $[109-i,73-i,16]_4 \hfill  0\le i \le 11$ & Example \ref{EN:3}\\
$4$ & $[109+i,73,16]_4 \hfill  1\le i \le 3$ & Example \ref{EN:3}\\
$4$ & $[110,71+i,16]_4 \hfill 0\le i\le 1$ & Example \ref{EN:3}\\
$4$ & $[111-i,57-i,25]_4  \hfill  0\le i \le 1 $& Example \ref{EN:6}\\
$4$ & $[111,72,16]_4$ & Example \ref{EN:3}\\
$4$ & $[114,57,26]_4$ & Example \ref{EN:6}\\
\hline
$5$ & $[87-i,44-i ,22]_5 \hfill 0\le i \le 2$ & Example \ref{EN:4}\\
$5$ & $[87-i, 43-i,23]_5 \hfill 0\le i \le 7$ & Example \ref{EN:4}\\
$5$ & $[87-i,42-i ,24]_5 \hfill 0\le i \le 6$ & Example \ref{EN:4}\\
$5$ & $[88,43,24]_5$ & Example \ref{EN:4}\\
$5$ & $[90,44,24]_5 $ & Example \ref{EN:4}\\
$5$ & $[101-i,75-i,13]_5 \hfill 0\le i \le 22$ & Example \ref{EN:4}\\
\hline
$7$ & $[58,45,8]_7$ & Example \ref{EN:5}\\
$7$ & $[59,44,9]_7$ & Example \ref{EN:5}\\
$7$ & $[60,44,10]_7$ & Example \ref{EN:5}\\
$7$ & $[74,38,20]_7$ & Example \ref{EN:5}\\          
$7$ & $[74,37,21]_7$ & Example \ref{EN:5}\\          
$7$ & $[74,36,22]_7$ & Example \ref{EN:5}\\          
$7$ & $[76,38,22]_7$ & Example \ref{EN:5}\\          
                        \hline
                  \end{tabular}%
        \caption{New record-breaking linear codes over $\F_q$}
              \label{T1}%
  \end{table}%

The next example discusses constacyclic codes over $\F_4$ that can be
used to construct new record-breaking binary quantum codes.  
\begin{example}\label{EN:6}
(1) Let $C$ be an $\omega$-constacyclic code of length $111$ over $\F_4$
  with defining set $Z(1) \cup Z(7) \cup Z(19)$. Our computations using
  Magma \cite{magma} show that $C$ is a $[111,57,25]_4$ Hermitian
  dual-containing linear code.  Computing the minimum distance took
  almost $70$ CPU years, and it was completed in about $47$ days using
  up to $600$ cores.  This code has a {\em better minimum distance}
  than the previously best-known linear code over $\F_4$, which has
  minimum distance $23$. Shortening the code $C$ gives another new
  code with parameters $[110,56,25]_4$.  Applying the construction of
  \cite[Theorem 2]{Calderbank} gives a {\em record-breaking} binary
  quantum code with parameters $[\![111,3,25]\!]_2$. Moreover,
  applying the construction of \cite[Corollary 3]{DastbastehLisonek}
  to the Hermitian dual of $C$ gives another {\em record-breaking}
  binary quantum code with parameters $[\![114,0,26]\!]_2$. The latter
  code is in correspondence to a Hermitian self-dual code over $\F_4$
  with parameters $[114,57,26]_4$, which is indeed another {\em
    record-breaking} linear code over $\F_4$.

  The $\omega$-constacyclic code of length $111$ over $\F_4$ with
  defining set $Z(1)\cup Z(7)\cup Z(19) \cup Z(37)$ is the Hermitian
  dual of the code $C$. It has parameters $[111,54,d]_q$ with
  $d\in\{26,28\}$.  Verifying that the minimum distance is indeed $28$
  is estimated to take about $33$ CPU years; we have not (yet) carried
  out this calculation, but the code $[111,54,28]_4$ would also
  improve the currently best parameters.
  
(2) Let $C$ be an $\omega$-constacyclic code of length $183$ over
  $\F_4$ with the defining set $Z(25)$. The code $C$ is Hermitian
  dual-containing linear code with parameters $[183,153,d\ge
    11]_4$. Establishing the lower bound $d\ge 11$ took about $5.4$
  CPU years and was completed in about $3.5$ days using up to $600$
  cores.  Verifying that the true minimum distance is $12$ is
  estimated to take about $200$ CPU years.  The construction of
  \cite[Theorem 2]{Calderbank} gives a {\em record-breaking} binary
  quantum code with parameters $[\![183,123,d\ge 11]\!]_2$.
\end{example}

Table \ref{T2} summarizes the parameters of our new binary quantum codes. 

\begin{table}[h!]
  \centering
      \begin{tabular}{|p{4.3cm}|p{2.6cm}|}
      \hline
 \centering{Parameters} & {\centering{Reference}} \\
          \hline
$[\![109,37,16]\!]_2 $ & Example \ref{EN:3}\\
$[\![111,3,25]\!]_2$ & Example \ref{EN:6}\\
$[\![114,0,26]\!]_2$ & Example \ref{EN:6}\\
$[\![183,123,d \ge 11]\!]_2 $ & Example \ref{EN:6}\\
                        \hline
                  \end{tabular}%
        \caption{New record-breaking binary quantum codes}
              \label{T2}%
  \end{table}%

It should be mentioned that applying the secondary constructions of
binary quantum codes discussed in \cite[Theorem 6]{Calderbank} to any
of the codes in Table \ref{T2} produces dozens of additional
record-breaking binary quantum codes and sometimes provides a method
for constructing many other binary quantum codes, the existence of
which was established by the Gilbert–Varshamov bound. In particular,
applying such secondary constructions to the code
$[\![109,37,16]\!]_2$ produces $90$ more record-breaking binary quantum
codes.

 \section*{Acknowledgement}
 The first and third author acknowledge support by the Spanish Ministry of Economy and Competitiveness through the MADDIE project (Grant No. PID2022-137099NB-C44), by the Spanish Ministry of Science and Innovation through the proyect “Few-qubit quantum hardware, algorithms and codes, on photonic and solid-state systems” (PLEC2021-008251), by the Ministry of Economic Affairs and Digital Transformation of the Spanish Government through the QUANTUM ENIA project call - Quantum Spain project, and by the European Union through the Recovery, Transformation and Resilience Plan - NextGenerationEU within the framework of the “Digital Spain 2026 Agenda”.

M.\,G. acknowledges support by the Foundation for Polish
Science (IRAP project, ICTQT, contracts no. 2018/MAB/5 and
2018/MAB/5/AS-1, co-financed by EU within the Smart Growth Operational
Programme). 

 \section*{Declaration of competing interest}
 The authors declare that there is no conflict of interest related to the research presented in the manuscript submitted to
Designs, Codes and Cryptography.

 \section*{Data Availibility}
Data sharing is not applicable to this article since no datasets were
generated or analyzed during the current study. Moreover, all
computational examples in this work are thoroughly explained through
relevant examples.

\bibliographystyle{abbrv}
\addcontentsline{toc}{chapter}{Bibliography}
\bibliography{EquivalenceReferences}

\end{document}